\documentclass{article} 
\usepackage[centertags,sumlimits,intlimits,namelimits,reqno]{amsmath}
\usepackage{latexsym,amsfonts,amssymb,exscale,enumerate,amsthm}
\usepackage[colorlinks,linkcolor=blue,citecolor=red,urlcolor=blue]{hyperref}

        \newcommand\R{{\mathbb R}}   
        \newcommand{\maps}{\colon}   
        \newcommand{\supp}{\mathrm{supp}}

	\newtheorem{theorem}{Theorem}

	\newtheorem{lemma}{Lemma}   
	   
       \newtheorem*{proposition*}{Proposition}

\begin{document}   
 
	\begin{center}   
	{\bf A Noether Theorem for Markov Processes \\}   
        \vspace{0.3cm}
	{\em John\ C.\ Baez \\}
        \vspace{0.3cm}
	{\small Centre for Quantum Technologies  \\
        National University of Singapore \\
        Singapore 117543  \\ and \\
        Department of Mathematics \\
        University of California \\
        Riverside CA 92521 \\ } 
        \vspace{0.3cm}
	{\em Brendan Fong \\}
        \vspace{0.3cm}
        {\small Mathematical Institute  \\
        University of Oxford  \\
        OX2 6UD, United Kingdom  \\ }
	\vspace{0.3cm}   
        {\small email:  baez@math.ucr.edu, fong@maths.ox.ac.uk\\} 
	\vspace{0.3cm}   
	{\small March 9, 2012}
	\vspace{0.3cm}   
	\end{center}

\begin{abstract}
\noindent Noether's theorem links the symmetries of a quantum system
with its conserved quantities, and is a cornerstone of quantum
mechanics. Here we prove a version of Noether's theorem for Markov
processes.  In quantum mechanics, an observable commutes with the
Hamiltonian if and only if its expected value remains constant in
time for every state.  For Markov processes that no longer holds, but
an observable commutes with the Hamiltonian if and only if both its
expected value and standard deviation are constant in time for 
every state.  
\end{abstract}

\section{Introduction}

There is a rich analogy between quantum mechanics and what one might
call `stochastic mechanics', where probabilities take the place of
amplitudes \cite{Azimuth}.  In quantum mechanics, we specify the state 
of a system by an element $\psi$ of a Hilbert space, and describe its 
time evolution by the Schr\"odinger equation: 
\[ 
 \frac{d}{dt} \psi = -iH\psi  
\]
where $H$ is a self-adjoint linear operator called the Hamiltonian.
For Markov processes, we specify the state of a system by a
probability distribution $\psi$ on some measure space, and describe its
time evolution by the so-called `master equation'
\[  
\frac{d}{dt} \psi = H\psi  
\]
where $H$ is a linear operator variously known as a `stochastic
Hamiltonian', `transition rate matrix' or `intensity matrix'.  
In quantum mechanics, it is well-known that conserved quantities 
correspond to self-adjoint operators that commute with the Hamiltonian.   
Here we present a similar result for Markov processes. 

To avoid technicalities and focus on the basic idea, we start by
considering Markov processes, or technically `Markov semigroups',
where the measure space $X$ is just a finite set equipped with its
counting measure.  Later we consider the general case.  We begin by
reviewing some basic facts and setting up some definitions; for
details see Stroock \cite{St}.

When $X$ is a finite set, a \textbf{probability distribution} on $X$
is a function $\psi \maps X \to \R $ such that $\psi_i \ge 0$ for
all $i \in X$ and
\[              \sum_i \psi_i = 1. \]
We say an operator $U \maps \R^X \to \R^X$ is \textbf{stochastic}
if it is linear and it maps probability distributions to probability
distributions.  A \textbf{Markov semigroup} consists of operators
$U(t) \maps \R^X \to \R^X$, one for each $t \in [0,\infty)$, such that:
\begin{enumerate}[(i)]
\item $U(t)$ is stochastic for all $t \ge 0$
\item $U(t)$ depends continuously on $t$.
\item $U(s+t) = U(s)U(t)$ for all $s,t \ge 0$.
\item $U(0) = I$.
\end{enumerate}
Any Markov semigroup may be written as $U(t) = \exp(tH)$ for a 
unique linear operator $H \maps \R^X \to \R^X$.  Moreover, this
operator $H$ is \textbf{infinitesimal stochastic}, meaning that if 
we write it as a matrix using the canonical basis for $\R^X$, then:
\begin{enumerate}[(i)]
\item $H_{ij} \ge 0$ for all $i,j \in X$ with $i \ne j$.
\item $\sum_{i \in X} H_{i j} = 0$ for all $i \in X$.
\end{enumerate}
For $i \ne j$ the matrix entry $H_{ij}$ is the probability per time of
a transition from the state $j\in X$ to the state $i \in X$.
Condition (i) says that these probabilities are nonnegative.
Condition (ii) then says that the diagonal entry $H_{ii}$ is minus the
probability per time of a transition out of the state $i$.

Conversely, for any infinitesimal stochastic operator $H \maps \R^X
\to \R^X$, $\exp(tH)$ is a Markov semigroup.  Given any function $\psi
\maps X \to \R$, we obtain a solution of the \textbf{master equation}:
\[  
\frac{d}{dt} \psi(t) = H\psi(t)
\]
with $H$ as \textbf{Hamiltonian} and $\psi$ as the initial value by setting
\[   \psi(t) = \exp(tH) \psi .\]
If $\psi$ is a probability distribution, then so is $\psi(t)$ for all
$t \ge 0$.

Next we turn to Noether's theorem.  There are many theorems of this
general type, all of which relate \emph{symmetries} of a physical
system to its \emph{conserved quantities}.  Noether's original theorem
applies to the Lagrangian approach to classical mechanics, and obtains
conserved quantities from symmetries of the Lagrangian \cite{By,KS}.  In the
Hamiltonian approach to classical mechanics, any observable having
vanishing Poisson brackets with the Hamiltonian both generates
symmetries of the Hamiltonian and is a conserved quantity.  This idea
extends to quantum mechanics if we replace Poisson brackets by
commutators.  It is this last form of Noether's theorem, somewhat removed
from the original form but very easy to prove, that we now generalize
to Markov processes.  For a Markov process, an observable will commute
with the Hamiltonian if and only if both its expected value and that
of its square are constant in time for every state.

Here an \textbf{observable} is a function $O \maps X \rightarrow \R$
assigning a real number $O_i$ to each state $i \in X$.  We identify $O$
with the diagonal matrix with $ii$th entry equal to $O$, and define
its \textbf{expected value} for a probability distribution $\psi$ to be
\[
\langle O, \psi \rangle =  \sum_{i \in X} O_i \psi_i  .
\]

Our Noether theorem for Markov processes may then be stated as follows:

\begin{proposition*}[Noether's Theorem, Stochastic Version]
Let $X$ be a finite set, let $H \maps \R^X \to \R^X$ be an
infinitesimal stochastic operator, and let $O$ be an observable. Then
$[O,H] = 0$ if and only if for all probability distributions $\psi(t)$
obeying the master equation $\frac{d}{dt} \psi(t) = H\psi(t)$, the
expected values $\langle O, \psi(t) \rangle$ and $\langle O^2,
\psi(t)\rangle$ are constant.
\end{proposition*}

For comparison, in the quantum version, both the Hamiltonian and
the observable are given by self-adjoint operators on a Hilbert
space.  To avoid technicalities, we only state the version for bounded
operators:

\begin{proposition*}[Noether's Theorem, Quantum Version]
Let $H$ and $O$ be bounded self-adjoint operators on a Hilbert space.
Then $[O,H] = 0$ if and only if for all states $\psi(t)$ obeying
Schr\"odinger's equation $\frac{d}{d t} \psi(t) = -i H \psi(t)$ the
expected value $ \langle \psi(t), O \psi(t) \rangle $ is constant.
\end{proposition*}

The similarity between these two results is striking, but this also
illuminates a key difference: the Markov version requires not just the
expected value of the observable to be constant, but also the expected
value of its square. This condition cannot be weakened to only require
that the expected value be constant. Observe that if
\[
H = \left(\begin{array}{rrr}
0 & 1 & 0 \\
0 & -2 & 0 \\
0 & 1 & 0
\end{array}\right) 
\qquad \rm{and} 
\qquad O=  \left(\begin{array}{ccc}
0 & 0 & 0 \\
0 & 1 & 0 \\
0 & 0 & 2 
\end{array}\right)
\]
then for $\psi(0) = (0,1,0)$, we have $\frac{d}{dt}\langle O, \psi
\rangle = 0$, but $[O,H] \ne 0$.

Indeed, in both the quantum and stochastic cases, the time derivative
of the expected value of an observable $O$ is expressed in terms of
the commutator $[O,H]$. In the quantum case we have
\[ 
\frac{d}{d t} \langle \psi(t), O \psi(t) \rangle = 
- i \langle \psi(t), [O,H] \psi(t) \rangle
\]
for any solution $\psi(t)$ of Schr\"odinger's equation.
The polarisation identity then implies that the right-hand side vanishes 
for all solutions if and only if $[O,H] = 0 $. In the stochastic case we 
have
\[
\frac{d}{d t} \langle O, \psi(t)\rangle 
= \langle 1, [O,H] \psi(t) \rangle
\]
for any solution $\psi(t)$ of the master equation.  However, in this
case the right-hand side can vanish for all solutions $\psi(t)$ without 
$[O,H] = 0$, as shown by the above example. To ensure $[O,H] = 0$ we 
need a supplementary hypothesis, such as the vanishing of 
$\frac{d}{d t} \langle O^2 ,\psi(t) \rangle$.   

What is the meaning of this supplementary hypothesis?  Including it
means that not only is the expected value of the observable $O$ 
conserved, but so is its \textbf{variance}, defined by
\[    \langle O^2 , \psi \rangle - \langle O, \psi \rangle^2 .\]
Of course the variance is the standard deviation of $O$, so 
an observable commutes with the Hamiltonian if and only if both its
expected value and standard deviation are constant in time for every state.

\section{Proof}

While proving Proposition 1 it is enlightening to introduce some other
equivalent characterizations of conserved quantities.  For this we
shall introduce the \textbf{transition graph} of an infinitesimal
stochastic operator.  Suppose $X$ is a finite set and $H \maps \R^X
\to \R^X$ is an infinitesimal stochastic operator.  We may form a
directed graph with the set $X$ as vertices and an edge from $j$ to
$i$ if and only if $H_{ij} \ne 0$.  We say $i$ and $j$ are in the same
\textbf{connected component} of this graph if there is a sequence of
vertices $j = k_0, k_1, \dots, k_n = i$ such that for each $0 \le \ell
< n$ there is either an edge from $k_\ell$ to $k_{\ell+1}$ or from
$k_{\ell+1}$ to $k_\ell$.

Our Noether theorem is the equivalence of (i) and (iii) in this result:

\begin{theorem} \label{thm1}
Let $X$ be a finite set, let $H \maps \R^X \to \R^X$ be an
infinitesimal stochastic operator, and let $O$ be an observable.  Then
the following are equivalent:

\begin{enumerate}[(i)]

\item $[O,H] =0$.

\item $ \frac{d}{d t} \langle f(O), \psi(t) \rangle = 0$ for all
polynomials $f\maps \R \rightarrow \R$ and all $\psi$ satisfying the
master equation with Hamiltonian $H$.

\item $ \frac{d}{d t} \langle O, \psi(t)\rangle = \frac{d}{d t} \langle
O^2, \psi(t)\rangle =0$ for all $\psi$ satisfying the master equation
with Hamtiltonian $H$.

\item $O_i = O_j$ if $i$ and $j$ lie in the same connected component
of transition graph of $H$.
\end{enumerate}
\end{theorem}

\begin{proof}
We prove (i) $\Rightarrow$ (ii) $\Rightarrow$ (iii) $\Rightarrow$ (iv)
$\Rightarrow$ (i).
\smallskip

\noindent \textbf{(i)\phantom{v} $\Rightarrow$ \phantom{i}(ii)} \quad
As $H$ commutes with $O$, the Taylor expansion of $f$ shows that $H$
commutes with $f(O)$ whenever $f$ is a polynomial. From this and the
master equation we have
\[
\tfrac{d}{d t} \langle f(O), \psi(t)\rangle 
= \langle f(O), \tfrac{d}{d t} \psi(t) \rangle 
= \langle f(O), H \psi(t)\rangle = \langle 1, f(O) H \psi(t)\rangle
= \langle 1, H f(O) \psi(t)\rangle
\]
But $H$ is infinitesimal stochastic, so
\[
\langle 1, H f(O) \psi(t) \rangle 
=\sum_{i, j \in X} H_{i j} f(O_j)\psi_j(t) 
= \sum_{j \in X} \bigg( \sum_{i \in X} H_{ij} \bigg) f(O_j)\psi_j(t) = 0 .
\]
\smallskip

\noindent \textbf{(ii)\phantom{'} $\Rightarrow$ (iii)} \quad Both $O$
and $O^2$ are polynomials in $O$.
\smallskip

\noindent \textbf{(iii) $\Rightarrow$ (iv)} \quad Suppose that $i,j
\in X$ lie in the same connected component.  We claim that then $O_i =
O_j$.  Clearly it suffices to show $O_i = O_j$ whenever $H_{ij} \ne 0$.
And for this, this it is enough to show that for any $j \in X$ we have
\[
\sum_{i \in X} (O_j-O_i)^2 H_{i j} =0.
\]
This is enough, as each term in this sum is nonnegative: when $i = j$
we have $O_j-O_i = 0$, while when $i \ne j$, both $(O_j-O_i)^2$ and
$H_{ij}$ are nonnegative---the latter because $H$ is infinitesimal
stochastic. Thus when their sum is zero each term $(O_j-O_i)^2H_{ij}$
is zero.  But this means that if $H_{ij}$ is nonzero, then $O_i =
O_j$, and this proves the claim.

Expanding the above expression then, we have
\[
\sum_{i\in X} (O_j-O_i)^2 H_{ij} = O_j^2\sum_{i \in X} H_{ij} - 2O_j
\sum_{i \in X} O_i H_{i j} + \sum_{i \in X} O_i^2 H_{i j}.
\]
The three terms here are each zero: the first because $H$ is
infinitesimal stochastic, and the latter two since, if $e_j$ is the
probability distribution with value 1 at $j \in X$ and 0 elsewhere,
then
\[
\left.\frac{d}{dt}\langle O, \exp(tH) e_j\rangle\right|_{t=0} = 
\langle O, He_j\rangle =  \sum_{i \in X} O_i H_{i j} 
\]
and
\[ \left.\frac{d}{dt}\langle O^2, \exp(tH) e_j\rangle\right|_{t=0} = 
\langle O^2, He_j\rangle = \sum_{i \in X} O_i^2 H_{i j},
\]
and by hypothesis these two derivatives are both zero.

\smallskip

\noindent \textbf{(iv)\phantom{i} $\Rightarrow$ \phantom{.}(i)} \quad
When $H_{i j}$ is nonzero, the states $i$ and $j$ lie in the same
component, so $O_i=O_j$. Thus for each $i,j \in X$:
\[
 [O,H]_{i j} = (O H - H O)_{i j} = O_i H_{i j} - H_{i j}O_j 
= (O_i-O_j)H_{i j}  = 0. \qquad \qedhere
\] 
\end{proof}

\section{Generalization}

In this section we generalize Noether's theorem for Markov processes
from the case of a finite set of states to a more general measure
space.  This seems to require some new ideas and techniques.  

Suppose that $X$ is a $\sigma$-finite measure space with a measure we
write simply as $dx$.  Then probability distributions $\psi$ on $X$
lie in $L^1(X)$.  We define an \textbf{observable} $O$ to be any element
of the dual Banach space $L^\infty(X)$, allowing us to define the expected
valued of $O$ in the probability distribution $\psi$ to be
\[   \langle O, \psi \rangle = \int_X O(x) \psi(x) \, dx .\]
We can also think of an observable $O$ as a bounded operator on 
$L^1(X)$, namely the operator of multiplying by the function $O$.

Let us say an operator $U \maps L^1(X) \to L^1(X)$ is \textbf{stochastic}
if it is linear, bounded, and maps probability distributions to probability
distributions.  Equivalently, $U$ is stochastic if it is linear and obeys
\[                \psi \ge 0 \implies U \psi \ge 0 \]
and
\[         \int_X (U\psi)(x) \, dx = \int_X \psi(x) \, dx \]
for all $\psi \in L^1(X)$.   We may also write the latter equation 
as 
\[    \langle 1, U \psi \rangle = \langle 1 , \psi \rangle .\]

A \textbf{Markov semigroup} is a strongly continuous one-parameter
semigroup of stochastic operators $U(t) \maps L^1(X) \to L^1(X)$.  By
the Hille--Yosida theorem \cite{Yo}, any Markov semigroup may be
written as $U(t) = \exp(tH)$ for a unique closed operator $H$ on
$L^1(X)$.  Any operator $H$ that arises this is \textbf{infinitesimal
stochastic}.  However, such operators are typically unbounded and only
densely defined.  This makes it difficult to work with the commutator
$[O,H]$, because the operator $O$ may not preserve the domain of $H$.
From our experience with quantum mechanics, the solution is to work
instead with the commutators $[O,\exp(tH)]$, which are bounded operators
defined on all of $L^1(X)$.  This amounts to working directly with the
Markov semigroup instead of the infinitesimal stochastic operator $H$.

\begin{theorem}
Suppose $X$ is a $\sigma$-finite measure space and 
$$U(t) \maps L^1(X) \to L^1(X)$$ 
is a Markov semigroup.  Suppose $O$ is an observable.  
Then $[O,U(t)] = 0$ for all $t \ge 0$ if and only if for all probability 
distributions $\psi$ on $X$, the expected values $\langle O, U(t) \psi 
\rangle$ and $\langle O^2, U(t) \psi\rangle$ are constant as a function of 
$t$.
\end{theorem}

This result is an easy consequence of the the following one, which is
of interest in its own right, since it amounts to a Noether's theorem
for Markov chains.  A `Markov chain' is similar to a Markov process,
but time comes in discrete steps, and at each step the probability
distribution $\psi$ evolves via $\psi \mapsto U \psi$ for some
stochastic operator $U$.

\begin{theorem}
Suppose $X$ is a $\sigma$-finite measure space and 
$U \maps L^1(X) \to L^1(X)$ is stochastic operator.  Suppose $O$ 
is an observable.  Then $[O,U] = 0$ if and only if for all 
probability distributions $\psi$ on $X$, $\langle O, U \psi 
\rangle = \langle O, \psi \rangle$ and $\langle O^2, 
U \psi \rangle = \langle O^2, \psi \rangle$. 
\end{theorem}

\begin{proof}
First, suppose $[O,U] = 0$.  Note 
\[    \langle O , \phi \rangle = \langle 1, O \phi \rangle \]
and since $U$ is stochastic, also
\[    \langle 1, U \phi \rangle = \langle 1, \phi \rangle \]
for all $\phi \in L^1(X)$.   Thus, for any 
probability distribution $\psi$ on $X$ and any $n \ge 0$ we
have
\[
   \langle O^n, U \psi \rangle 
= \langle 1, O^n U \psi \rangle  
= \langle 1, U O^n \psi \rangle  
= \langle 1, O^n \psi \rangle 
= \langle O^n, \psi \rangle . 
\]
Taking $n = 1,2$ we get the desired result.

To prove the converse, we use three lemmas.  In all these $X$ is a
$\sigma$-finite measure space, $U \maps L^1(X) \to L^1(X)$ is a
stochastic operator, and $O$ is an observable.  We freely switch
between thinking of $O$ as a function in $L^\infty(X)$ and the
operator on $L^1(X)$ given by multiplying by that function.

\begin{lemma}
Suppose that for any compact interval $I \subseteq \R$ the operator
$U$ commutes with
$\chi_{I}(O)$, meaning the operator on $L^1(X)$ given by
multiplying by the characteristic function of
\[     O^{-1}(I) = \{ x \in X \colon \; O(x) \in I \} . \]
Then $U$ commutes with $O$.
\end{lemma}

\begin{proof}
The range of the function $O$ is contained in the interior of some interval
$[-M,M]$.  The step functions
\[        f_n = \sum_{i=-n}^{n-1} \frac{i M}{n} \, 
\chi_{[\frac{i M}{n}, \frac{(i+1)M}{n})} \]
are uniformly bounded and converge pointwise to the identity function
on the range of $O$, so by the dominated convergence theorem $f_n(O)
\psi \to O \psi$ in the $L^1$ norm for all $\psi \in L^1(X)$.  Furthermore,
though we have not written as such, $f_n$ is a linear combination of 
characteristic functions of compact intervals, since we include a single
point as a degenerate special case of a compact interval.   Thus, by 
hypothesis, $U$ commutes with $f_n(O)$.  It follows that
for every $\psi \in L^1(X)$, we have
\[    
\begin{array}{ccl}
O U \psi &=& \lim_{n \to \infty} f_n(O) U \psi   \\
&=& \lim_{n \to \infty} U f_n(O) \psi \\
&=& U O \psi .  
\end{array} \qedhere
\]
\end{proof}

\begin{lemma}
Suppose that for every compact interval $I$ and every $\psi \in L^1(X)$,
\[   \supp (\psi) \subseteq O^{-1}(I) \; \Longrightarrow \; 
\supp U (\psi) \subseteq O^{-1}(I) .  \]
Then $U$ commutes with every operator $\chi_I(O)$.
\end{lemma}

\begin{proof}
The set of $L^1$ functions supported in $O^{-1}(I)$ is the range of
the operator $\chi_I(O)$, so the hypothesis says that $U$ maps the 
range of this operator to itself.  
Given any $\psi \in L^1(X)$ and writing $p = \chi_I(O)$, we have
\[  \psi = p \psi + (1-p) \psi \]
so
\[   p U \psi = p U p \psi + p U (1-p) \psi .\]
Since $U$ preserves the range of $p$ 
and $p$ is the identity on this range, we have $p U p \psi = U p \psi$.
Since the range of the function $O$ is contained in some interval
$[-M,M]$, can write $1-p$ as a linear combination of operators
$\chi_J(O)$ for other compact intervals $J$, again using
the fact that a point is a degenerate case of a compact interval.
Thus $U$ also preserves the range of $1-p$.  Since the range of 
$1-p$ is the kernel of $p$, $pU(1-p) \psi = 0$.  We thus have
\[   p U \psi = U p \psi.  \qedhere \]
\end{proof}

\begin{lemma} Suppose that
$\langle O, U \psi \rangle = \langle O, \psi \rangle$ and 
$\langle O^2, U \psi \rangle = \langle O^2, \psi \rangle$
for all $\psi \in L^1(X)$.  
Then for every compact interval $I$ and every $\psi \in L^1(X)$,
\[   \supp( \psi) \subseteq O^{-1}(I) \; \Longrightarrow \; 
\supp (U \psi) \subseteq O^{-1}(I) .  \]
\end{lemma}

\begin{proof} 
The range of $O$ is contained in the interior of some interval
$[-M,M]$.  Thus we shall only prove the lemma for $I$ contained in
$(-M,M)$, since otherwise we can replace $I$ by a smaller compact
interval with this property without changing $O^{-1}(I)$.

Suppose $\psi$ is supported in $O^{-1}(I)$.  We wish to show the same
for $U\psi$.  It suffices to show that $U\psi$ is supported in $O^{-1}(J)$
where $J$ is any compact interval with $I \subseteq \mathrm{int}(J) \subseteq
(-M,M)$.  Moreover, we may assume that $\psi$ is a probability
distribution, since any $L^1$ function supported in $O^{-1}(I)$ is a
linear combination of two probability distributions supported in this
set.

To show $U\psi$ is supported in $J$, we write
\[   \psi = \sum_{i=-n}^{n-1} \psi_i  \]
where 
\[   \psi_i = \chi_{[\frac{i M}{n}, \frac{(i+1)M}{n})}(O)\, \psi .\]
We thus have
\[   U \psi = \sum_{i=-n}^{n-1} U \psi_i . \]
We shall show that 
\begin{equation}
\label{inequality}
       \| (1 - \chi_J(O)) U \psi_i \|_1  \le  \frac{c}{n^2}\, \|\psi_i\|_1  
\end{equation}
for some constant $c$.  
It follows that
\[ 
 \| (1 - \chi_J(O)) U\psi \|_1 \le
 \sum_{i=-n}^{n-1} \| (1 - \chi_J(O)) U \psi_i \|_1 \le
 \frac{c}{n^2} \sum_{i=-n}^{n-1} \|\psi_i \|_1
= \frac{c}{n^2} \, \|\psi\|_1 = \frac{c}{n^2}
\]
for all $n$, and thus
\[   \| (1 - \chi_J(O)) U\psi \|_1 = 0 \]
so $U\psi$ is supported in $O^{-1}(J)$.

To prove \eqref{inequality} we shall use Chebyshev's inequality, which
says that the probability of a random variable taking a value
at least $k$ standard deviations away from its mean is less than
or equal to $1/k^2$.   However, we first need to convert probability
functions on $X$ into probability measures on the real line.

For any function $\psi \in L^1(X)$ we can push forward the 
signed measure $\psi \, dx$ on $X$ via $O \maps X \to \R$ to obtain a 
signed measure on $\R$, which we call $\widetilde{\psi}$.  Concretely,
$\tilde{\psi}$ is characterized by the equation
\[  
\int_X f(O(x)) \psi(x) \, dx = 
\int_\R f \, \widetilde{\psi}
\]
which holds for any bounded measurable function $f \maps \R \to \R$.
Since the integral at left only depends on the value of $f$ in the
interval $[-M,M]$, the same is true for the integral at right, so
with no harm we can restrict $f$ and the integral to $[-M,M]$.  Clearly 
\[ \widetilde{(\phi + \chi)} = \widetilde{\phi} + \widetilde{\chi} \]
and
\[   \widetilde{(\alpha \phi)} = \alpha \, \widetilde{\phi} \]
for all $\alpha \in \R$ and $\phi, \chi \in L^1(X)$.  
If $\phi$ is nonnegative then $\widetilde{\phi}$ is a nonnegative
measure, and if $\phi$ is a probability distribution then 
$\widetilde{\phi}$ is a probability measure.

Since we are assuming the function $\psi$ is supported in $O^{-1} I$,
the measure $\widetilde{\psi}$ is supported in $I$.  Since
\[   \psi_i = \chi_{[\frac{i M}{n}, \frac{(i+1)M}{n})}(O) \psi  \]
we have 
\[  \widetilde{\psi_i} = \chi_{[\frac{i M}{n}, \frac{(i+1)M}{n})} 
\, \widetilde{\psi}  .\]
In what follows we assume $\psi_i$ is nonzero, since otherwise
\eqref{inequality} is trivial.  This implies that we can 
rescale $\widetilde{\psi_i}$ to obtain a probability 
measure, namely 
\[     \frac{1}{\|\psi_i\|_1} \, \widetilde{\psi_i}.\]
Since this probability measure is supported in the interval
$I \cap [iM/n, (i+1)M/n)$, clearly its mean lies in $I$, and its
standard deviation is $\le M/n$.

By our hypotheses on $U$,
\[   \frac{1}{\|\psi_i\|_1} \, \widetilde{U \psi_i} \]
is another probability measure on the real line, with the same
mean and standard deviation as the one above.  Thus, by Chebyshev's
inequality, the integral of this probability measure over the 
complement of $J$ is less than or equal to $((M/n)/d)^2$, where
$d$ is the distance of $I$ from the complement of $J$.  In other
words,
\[    \frac{1}{\|\psi_i\|_1} \,  \int_{\R - J} \widetilde{U \psi_i}
\le \left(\frac{Md }{n}\right)^2 \]
or writing $(M d)^2 = c$, 
\[     \int_\R (1 - \chi_J) \, \widetilde{U\psi_i}
\le \frac{c}{n^2} \, \|\psi_i \|_1 \]
or equivalently
\[     \int_X (1 - \chi_J(O)(x)) \, (U\psi_i)(x) \, dx 
\le \frac{c}{n^2} \, \|\psi_i \|_1 .\]
Since the integrand is nonnegative, this implies
\[
      \| (1 - \chi_J(O)) U \psi_i \|_1  \le  \frac{c}{n^2} \|\psi_i\|_1  
\]
which is \eqref{inequality}, as desired.  
\end{proof}

Combining these three lemmas, the converse follows.  \end{proof}

\subsection*{Acknowledgements}
BF's research was supported by an internship at the Centre for 
Quantum Technologies.

\end{document}